\documentclass[oneside,a4paper,11pt,reqno]{amsart}

\usepackage[margin=35mm]{geometry}
\usepackage[utf8]{inputenc}
\usepackage{mathrsfs,dsfont,amsmath,amssymb,amsthm,amsfonts,amstext,amsopn,amsxtra,mathrsfs,esint,enumitem,bookmark,mathtools}
\usepackage[capitalize]{cleveref}
\usepackage{braket}
\usepackage{amssymb}
\usepackage{stmaryrd}

\newtheorem{lemma}{Lemma}
\newtheorem{theorem}[lemma]{Theorem}
\newtheorem{proposition}[lemma]{Proposition}

\theoremstyle{definition}
\newcommand\xqed[1]{%
	\leavevmode\unskip\penalty9999 \hbox{}\nobreak\hfill\quad\hbox{#1}%
}
\newcommand\remarkend{\xqed{$\triangle$}}
\makeatletter
    \def\@endtheorem{\remarkend\endtrivlist\@endpefalse }
\makeatother
\theoremstyle{remark}
\newtheorem{remark}[lemma]{Remark} 
\makeatletter
    \def\@endtheorem{\endtrivlist\@endpefalse }
\makeatother

\renewcommand\phi{\varphi}

\renewcommand{\leq}{\leqslant}

\newcommand{\R}{\mathbb{R}}

\newcommand{\C}{\mathbb{C}}
\newcommand\1{{\ensuremath {\mathds 1} }}
\newcommand{\cM}{\mathcal{M}}
\newcommand{\tr}{{\rm Tr}\,}
\newcommand{\supp}{{\rm Supp}}

\newcommand{\bx}{\mathbf{x}}
\newcommand{\by}{\mathbf{y}}
\newcommand{\bz}{\mathbf{z}}

\newcommand{\dd}{\mathrm{d}}
\def\d{\,{\rm d}}

\usepackage{hyperref}

\title[The low density Bose gas with three-body interactions]{Ground state energy of the low density Bose gas with three-body interactions}

\author[P.T. Nam]{Phan Th{\`a}nh Nam}
\address{Department of Mathematics, LMU Munich, Theresienstrasse 39, 80333 Munich, and Munich Center for Quantum Science and Technology, Schellingstr. 4, 80799 Munich, Germany} 
\email{nam@math.lmu.de}

\author[J. Ricaud]{Julien Ricaud}
\address{Centre de Math\'{e}matiques Appliqu{\'e}es, {\'E}cole polytechnique, 91128 Pa\-lai\-seau Cedex, France}
\email{julien.ricaud@polytechnique.edu}

\author[A. Triay]{Arnaud Triay}
\address{Department of Mathematics, LMU Munich, Theresienstrasse 39, 80333 Munich, and Munich Center for Quantum Science and Technology, Schellingstr. 4, 80799 Munich, Germany}  
\email{triay@math.lmu.de}

\begin{document}

\begin{abstract}
	We consider the low density Bose gas in the thermodynamic limit with a three-body interaction potential. We prove that the leading order of the ground state energy of the system is determined completely in terms of the scattering energy of the interaction potential. The corresponding result for two-body interactions was proved in seminal papers of Dyson (1957) and of Lieb--Yngvason (1998).
\end{abstract}

\maketitle

\begin{center}
\emph{Dedicated to the memory of Freeman J.~Dyson (1923-2020)}
\end{center}

\bigskip

\section{Introduction} 

The Bose--Einstein condensation (BEC) is the phenomenon where many bosonic particles occupy a common one-body quantum state. It was predicted in 1924~\cite{Bose-24,Einstein-24,Einstein-25} and experimentally observed in 1995~\cite{Wieman-Cornell-95, Ketterle-95}, but the rigorous derivation of the BEC from first principles remains a major question in quantum physics. In fact, the pioneer works of Bose~\cite{Bose-24} and Einstein~\cite{Einstein-24,Einstein-25} are rigorous, but they concern only the non-interacting gas. The interactions between particles are essential to explain several phenomena such as superfluidity~\cite{Landau-41} and quantized vortices~\cite{Onsager-49}, but they complicate the analysis dramatically.

In general, a genuine many-body interaction potential of the form $U(x_1,\dots,x_N)$, where $N$ is the number of particles, is too difficult for practical computations. In dilute Bose gases, which are most relevant to the experiments in~\cite{Wieman-Cornell-95, Ketterle-95}, the range of the interaction is much smaller than the average distance between particles, and hence the interaction is often described by an effective potential depending only on few variables. 
Due to its simplicity, the two-body interaction is most assumed in the literature. In this context, the mathematical theory of interacting Bose gases goes back to Bogoliubov's 1947 paper~\cite{Bogoliubov-47b} where excited particles (the particles outside of the condensate) are treated as if they were quasi-free, leading to a prediction of the ground state energy and the excitation spectrum. In particular, Bogoliubov's theory gives a qualitative explanation of Landau's criterion for superfluidity~\cite{Landau-41}. However, as already noticed in~\cite{Bogoliubov-47b}, when applied to dilute Bose gases, Bogoliubov's approximation does not capture correctly the two-body scattering process of particles. Mathematically, this means that the usual mean-field approximation admits a subtle correction due to the correlation between particles.

While the emergence of the scattering length can be heuristically derived using perturbation methods~\cite{HY-57,LHY-57}, the rigorous understanding from first principles is highly nontrivial. In a seminal paper in 1957~\cite{Dyson-57}, Dyson proved rigorously that the ground state energy per volume in the thermodynamic limit satisfies 
\begin{equation} \label{eq:Dyson}
	C_{\rm Dys} a \rho^2 (1 + o(1)_{\rho a^3 \to 0}) \le e_{\rm 2B}(\rho) \le 4\pi a \rho^2 (1 + o(1)_{\rho a^3 \to 0})\,.
\end{equation} 
Here $\rho$ is the density of the system and $a$ is the scattering length of the two-body interaction; the condition $\rho a^3\to 0$ places us in the dilute regime. In~\cite{Dyson-57}, Dyson focused on a hard-sphere gas, but his argument can be translated to include general, positive potentials of finite range. Thus the significance of~\eqref{eq:Dyson} is the {\em universality}, namely the leading order of the complicated many-body energy can be determined in terms of only the scattering length of the interaction (any other details of the interaction potential is irrelevant). It turns out that the upper bound in~\eqref{eq:Dyson} is sharp, while the lower bound is about 14 times smaller than the correct one. It took some 40 years until Lieb--Yngvason~\cite{LieYng-98} proved the matching lower bound, thus concluding 
\begin{equation} \label{eq:LY}
	e_{\rm 2B}(\rho) = 4\pi a \rho^2 (1 + o(1)_{\rho a^3 \to 0})\,.
\end{equation}
The proof in~\cite{LieYng-98} also uses Dyson's important idea from~\cite{Dyson-57} of substituting a soft potential for the original one by sacrificing the kinetic energy. This argument, often referred to as {\em Dyson's lemma}, plays an important role in various dilute models, e.g. the Gross--Pitaevskii limit studied in~\cite{LieSeiYng-00,LieSei-02,LieSei-06,NamRouSei-16}. See~\cite{YauYin-09,BCS-21,FouSol-20,FouSol-21} for rigorous results on the next order correction to~\eqref{eq:LY}.

Although the two-body interaction is enough for many applications, in some cases the three-body correction is not negligible~\cite{BMZ-07,Petrov-14}. They contribute significantly, for example, to the computation of the binding energy of water~\cite{Mas-03}. In ultracold quantum gases, they can also be artificially enhanced by external fields, in a similar way to how the two-body scattering length is tuned by Feshbach resonance, and are expected to give rise to exotic physics like Pfaffian states~\cite{Petrov-14}. On the mathematical side, it is unclear how Bogoliubov's approximation should be modified in the dilute regime. In particular, to our knowledge, there is still a notable absence of heuristic discussion on the emergence of the three-body scattering process, let alone the rigorous understanding from first principles. The time-dependent problem, with the mean-field type potential $N^{6\beta-2} V\!\left( N^{\beta}(x-y,x-z) \right)$ for $\beta \ge 0$ small, has already been studied~\cite{ChePav-11,Chen-12,Yuan-15,CheHol-19,Lee-20,NamSal-20,LiYao-21}, and in a recent work~\cite{NamRicTri-21} we derived the leading order of the ground state energy in the Gross--Pitaevskii limit $\beta = 1/2$. The main purpose of the present paper is to extend the analysis to the thermodynamic limit, thus proving an analogue of~\eqref{eq:LY} for the low density Bose gases with three-body interactions. 

\bigskip
\noindent 
{\bf Acknowledgments.} We received funding from the Deutsche Forschungsgemeinschaft (DFG, German Research Foundation) under Germany's Excellence Strategy (EXC-2111-390814868). J.R. also acknowledges financial support from the French Agence Nationale de la Recherche (ANR) under Grant~No.~ANR-19-CE46-0007 (project ICCI).

\section{Main result}

\subsection{Model}
We consider $N$ bosons in $\Omega = [-L/2,L/2]^{3}$ for some $L>0$, interacting via a non-negative potential $V: \mathbb{R}^{3}\times \mathbb{R}^{3} \to [0,\infty)$. The system is described via the Hamiltonian
\begin{equation}\label{eq:H_NL}
	H_{N,L} = \sum_{i=1}^N -\Delta_{x_i} + \sum_{1\le i < j < k \le N} V(x_i - x_j,x_i-x_k)
\end{equation}
acting on the bosonic space $L^2_s(\Omega^N)$ where $-\Delta$ denotes the Laplacian with Neumann boundary conditions on $\Omega$. Since $H_{N,L}$ has to let $L^2_s(\Omega^N)$ invariant, this imposes the following \emph{three-body symmetry} on the interaction potential
\begin{equation}\label{eq:sym}
	V(x,y) = V(y,x) \quad \textrm{ and } \quad V(x-y,x-z) = V(y-x,y-z) = V(z-y,z-x)\,.
\end{equation}
The thermodynamic ground state energy per volume is defined as
\begin{equation} \label{eq:erho}
	e_{\rm 3B}(\rho) := \lim_{\substack{N \to \infty \\ N / L^3 \to \rho}} \inf_{\|\Psi\|_{L^2}^2 = 1} \frac{\langle \Psi, H_{N,L} \Psi \rangle}{L^3}\,.
\end{equation}
That the limit exists and does not depend on the boundary conditions is well-known, see for instance~\cite{Ruelle}. We will estimate $e_{\rm 3B}(\rho)$ in terms of the scattering energy of the interaction potential $V$. 

\subsection{Scattering energy}
	\label{sec:scat_energy}
Let $d\ge 3$ and $0\le v \in L^\infty(\R^d)$ be compactly supported. We define the {\em zero-scattering energy} of $v$ by 
\begin{equation} \label{eq:def-scat-energy}
	b(v):= \inf_{\varphi \in \dot H^1(\R^d)} \int_{\R^d} \left( 2|\nabla \varphi( \bx)|^2 + v(\bx) |1-\varphi(\bx)|^2 \right) \d \bx\,.
\end{equation}
Here $\dot H^1(\R^d)$ is the space of functions $g:\R^d \to \C$ vanishing at infinity with $|\nabla g|\in L^2(\R^d)$. Equivalently, we can also write $b(v)=\lim_{R\to \infty} b_R(v)$ where $b_R(v)$ is the energy in the ball $B(0,R)$ with the boundary condition $\varphi\equiv0$ on $\{|x|=R\}$ (see~\cite[Appendix C]{LSSY} for the latter definition). Contrarily to~\cite{LSSY} however, we do not need to assume that $v$ is radially symmetric. As proved in~\cite{NamRicTri-21}, the variational problem~\eqref{eq:def-scat-energy} has an optimizer $\omega=1-f$ where $f$ solves the scattering equation
\[
	-2\Delta f (\bx) + v(\bx) f(\bx) = 0\,, \quad \forall \bx\in \R^d \quad \textrm{ and } \quad \lim_{|\bx|\to \infty} f(\bx)=1\,.
\]
Integrating the above equation against $\omega(\bx) = 1 - f(\bx)$, one obtains an alternative expression for the modified scattering energy 
\[
	b(v) = \int_{\mathbb{R}^{d}} v(\bx) f(\bx)\, \d \bx\,.
\]

Note that if we formally insert the hard-sphere potential 
\[
	v_{\rm hs}(\bx) =
	\begin{cases}
		\infty, &\quad |\bx|<a\\
		0, &\quad |\bx|>a 
	\end{cases}
\]
in~\eqref{eq:def-scat-energy}, then we find that $b(v_{\rm hs})=c_d a^{d-2}$. For example $b(v_{\rm hs})=8\pi a$ when $d=3$. Thus $b(v)^{1/(d-2)}$ plays the role of the scattering length, up to a universal factor. In particular, for three-dimensional particles with three-body interactions ($d=6$), the~limit 
\[
	\rho b(v)^{3/4} \to 0
\]
corresponds to the dilute regime, where the length of the interaction ($\sim b(v)^{1/4}$) is much smaller than the mean distance between particles ($\sim \rho^{-1/3}$). 

\subsection{Modified scattering energy.}
We now focus on the Hamiltonian in~\eqref{eq:H_NL}. It turns out that the leading order of $e_{\rm 3B}(\rho)$ will be given in terms of a modified scattering energy of $V$, instead of the usual one as in~\eqref{eq:def-scat-energy}. Introducing $\cM: \R^3\times \R^3 \to \R^3\times \R^3$ given by
\begin{equation}\label{eq:intro-M}
	\cM
	= \frac{1}{2\sqrt{2}}
	\begin{pmatrix}
		\sqrt{3}+1 & \sqrt{3}-1 \\
		\sqrt{3}-1 & \sqrt{3}+1 
	\end{pmatrix},
\end{equation}
we define the modified scattering length of $V$ as
\[
	b_{\cM}(V) : = b(V(\cM \cdot)) \det \cM \,.
\]
Equivalently, we can express it similarly to~\eqref{eq:def-scat-energy} as 
\begin{equation} \label{eq:def_b}
	b_{\cM}(V) = \inf_{\varphi \in \dot H^1(\R^6)} \int_{\R^d} \left( 2| \cM \nabla \varphi(\bx)|^2 + V(\bx) |1-\varphi(\bx)|^2 \right) \d \bx\,.
\end{equation}
The matrix $\cM$ naturally appears when the problem~\eqref{eq:def_b} is posed in terms of reduced coordinates. Indeed, the change of coordinates
\[
	r_{1} = \frac{1}{3}(x_1+x_2+x_3)\,, \quad r_2 = x_1 -x_2\,, \quad \textrm{ and } \quad r_3 =x_1-x_3\,,
\]
leads to
\begin{multline}\label{eq:remove-center}
	-\Delta_{x_1} - \Delta_{x_2} - \Delta_{x_3} + V(x_1-x_2,x_1-x_3)\\
	\begin{aligned}[b]
		&= \left( \frac{1}{3} p_{r_1} + p_{r_2} + p_{r_3} \right)^2 + \left(\frac{1}{3} p_{r_1} - p_{r_2} \right)^2 + \left( \frac{1}{3} p_{r_1}-p_{r_3} \right)^2 + V(r_1,r_2) \\
		&= \frac{1}{3}p_{r_1}^2 + 2(p_{r_2}^2 + p_{r_3}^2 + p_{r_2} p_{r_3} ) + V(r_2,r_3)\,,
	\end{aligned}
\end{multline}
where we have denoted $p_x= -{\bf i}\nabla_x$. The last two terms, which are independent of the first one, lead to the minimization problem~\eqref{eq:def_b}. The matrix $\cM$ has been chosen so that 
\begin{equation}\label{eq:M2}
	\cM^2= \frac{1}{2}
	\begin{pmatrix}
		2 & 1 \\
		1 & 2
	\end{pmatrix}.
\end{equation}
As proved in~\cite{NamRicTri-21}, the variational problem~\eqref{eq:def_b} has an optimizer $\omega=1-f_{\cM}$ where $f_{\cM}$ satisfies the three-body symmetry~\eqref{eq:sym} and solves the scattering equation
\[
	-2\Delta_{\cM} f_{\cM} (\bx) + V(\bx) f_{\cM}(\bx) = 0\,, \quad \forall \bx\in \R^6 \quad \textrm{ and } \quad \lim_{|\bx|\to \infty} f(\bx)=1\,,
\]
with $-\Delta_{\cM}= |\cM \cdot \nabla|^2$. Then the modified scattering energy can be written as 
\[
	b_{\cM}(V) = \int_{\mathbb{R}^{6}} V(\bx) f_{\cM}(\bx)\, \d \bx \,.
\]
Since $b_{\cM}(V)$ and $b(V)$ have the same order of magnitude, $b_{\cM} (V)^{1/4}$ is proportional to the length of the interaction and the diluteness is still encoded in the limit $\rho b_{\cM} (V)^{3/4} \to 0$ (see also Section~\ref{sec:scat_energy}).

\subsection{Main result}
We can now state our main theorem.
\begin{theorem}\label{theo1}
	Let $V:\R^3 \times \R^3 \to [0,\infty)$ be bounded, compactly supported namely $\supp \, V \subset \{|\bx| \le R_0\}$ for some $R_0 >0$, and satisfy the three-body symmetry~\eqref{eq:sym}. 	Then in the dilute limit $Y:=\rho b_{\cM} (V)^{3/4} \to 0$, the thermodynamic ground state energy per volume in~\eqref{eq:erho} satisfies 
	\[
		e_{\rm 3B}(\rho) = \frac{1}{6}b_{\mathcal M}(V) \rho^3 (1 + \mathcal O(Y^{\nu}))
	\]
	for some constant $\nu>0$. 
\end{theorem}

\begin{remark}
	Our proof of the lower bound is easily extendable to hard-core potentials, and the error $\mathcal O(Y^{\nu})$ is uniform in $V$ as long as its range stays bounded or even increases slowly. For the upper bound, however, the error is uniform in $V$ assuming that $R_0/ b_{\cM}(V)^{1/4}, \|V\|_{L^1}b_{\cM}(V)^{-1}$ and $ \|V\|_{L^\infty}b_{\cM}(V)^{1/2}$ stay bounded. Those conditions come from the error estimate in Lemma~\ref{lem:upper_bound} and could be easily improved, see also the end of Section~\ref{sec:UB}. We believe the result to be true for hard-core potentials too, as Dyson proved it for two-body interactions, but we are unfortunately not able to adapt Dyson's analysis of the upper bound and our proof only works for bounded potentials. 
\end{remark}

The proof of Theorem~\ref{theo1} occupies the rest of the paper. In Section~\ref{sec:scattering}, we recall a Dyson's lemma for the three-body interaction potential. This ingredient is similar to that of~\cite{NamRicTri-21}, except that in the present paper we have to take into account the boundary condition carefully since we will apply the lemma to bounded sets. In Section~\ref{sec:LB}, we prove the lower bound $e_{\rm 3B}(\rho) \ge (1/6)b_{\mathcal M}(V) \rho^3 (1 + \mathcal O(Y^\nu))$ by following the localization method of Lieb--Yngvason~\cite{LieYng-98}; more precisely we will use a many-body version of Dyson's lemma and the Temple inequality on small boxes. The upper bound $e_{\rm 3B}(\rho) \le (1/6)b_{\mathcal M}(V) \rho^3 (1 + \mathcal O(Y^\nu))$ essentially follows from our analysis in the Gross--Pitaevskii limit in~\cite{NamRicTri-21}, and the details will be explained in Section~\ref{sec:UB}.

\section{Dyson's lemma} \label{sec:scattering}

\begin{lemma}[Dyson Lemma for non-radial potentials]\label{lem:dyson_lemma}
	Let $d\ge 3$, $R_2 / 2 > R_1 > R_0 > 0$ and $\{|\bx| \le R_2\} \subset \Omega \subset \mathbb{R}^{d}$ be an open set. Let $0\le v \in L^\infty(\R^d)$ with $\supp\, v \subset \{ |\bx| \le R_0\}$. Then there exists $0\le U \in C(\R^d)$ with $\supp \, U \subset \{ R_1\le |\bx| \le R_2\}$ and $\int_{\R^d} U = 1$ such that the following operator inequality holds on $L^2(\Omega)$
	\[
		- 2 \cM\nabla_{\bx} \1_{\{|\bx| \le R_2\}} \cM\nabla_{\bx}+ v(\bx) \ge b_{\cM}(v) \left( 1- \frac{C_d R_0}{R_1}\right) U (\bx)\,,
	\]
	with a constant $C_d>0$ depending only on the dimension $d$. 
\end{lemma}
\begin{proof}
	The proof is a simple extension of the one in~\cite{NamRicTri-21} to general domains, we sketch the main steps. 
	
	\medskip
	\textit{Step 1: Removing the $\cM$ dependence.} We see that by the change of variable $\bx =\cM {\bf y}$, it is enough to prove that under the same assumptions
	\[
		- 2 \nabla_{\by} \1_{\{|\cM \by| \le \sqrt{2/3} R_2\}} \nabla_{\by}+ v(\cM \by) \ge b(v(\cM \cdot )) \left( 1- \frac{C_d R_0}{R_1}\right) U (\cM \by)\quad \text{ on }L^2(\Omega)\,.
	\]
	We assume $R_1 > 2 R_0$ without loss of generality as, otherwise, the claim holds with~$C_d = 2$. Using that $\sqrt{1/2} \le \cM \le \sqrt{3/2}$, we obtain $\{|\by| \le \sqrt{2/3} R_2\} \subset \{|\cM \by| \le R_2\}$ and $\supp \, v(\cM \cdot) \subset \{|\bx| \le \sqrt{2} R_0\}$. So it is sufficient to prove the original claim for $\cM = 1$, $R_0' = \sqrt 2 R_0$, $R_1' = \sqrt 2 R_1$ and $R_2' = \sqrt{2/3} R_2$.
	
	\medskip
	\textit{Step 2: The case $\cM = 1$.} We abuse notation and omit the prime in $R_i'$ for $0\le i \le 2$. As proved in~\cite{NamRicTri-21}, there exists a solution $f$ to the scattering equation 
	\[
		0< f(\bx) \le 1, \quad -2\Delta f(\bx) + v(\bx) f(\bx) = 0\,, \quad \forall \bx\in \R^6\,, \quad \textrm{ and } \quad \lim_{|\bx|\to \infty} f(\bx) =1\,.
	\]
	Then for any $\phi \in C^\infty(\Omega)$ compactly supported, we define $\eta = \phi / f$, and for any $R \in [R_1,R_2]$, using the equation satisfied by $f$, we obtain
	\begin{multline*}
		\int_{B(0,R_2)} \left( 2|\nabla \varphi|^2 + v |\varphi|^2 \right) \\
		\ge \int_{B(0,R)} \left( 2|\nabla \varphi|^2 + v |\varphi|^2 \right) = \int_{B(0,R)} 2 f^2 |\nabla \eta |^2 + \int_{\partial B(0,R)} 2 |\varphi|^2 f^{-1} \nabla f \cdot \vec n \\
		\ge 2 \int_{\partial B(0,R)} |\varphi|^2 f^{-1}\nabla f \cdot \vec n\,,
	\end{multline*}
	where $\vec n_{\bx}= \bx/|\bx|$ is the outward unit normal vector on the sphere $\partial B(0,R)$. Now using that $|\mathbb{S}^{d-1}|^{-1} |\bx|^{d-1}$ is the Green function of the Laplacian on $\mathbb{R}^{d}$, we obtain
	\[
		\nabla f(\bx) \cdot \vec n_{\bx} \ge \frac{1}{2 |\mathbb{S}^{d-1}|} \int_{\R^d} \frac{v({\bf y}) f({\bf y})}{|\bx|^{d-1}} \left( 1 - C_d \frac{R_0}{|\bx|} \right) \d {\bf y} = \frac{b(v)}{ 2 |\mathbb{S}^{d-1}| |\bx|^{d-1}} \left( 1- \frac{C_d R_0}{|\bx|}\right). 
	\]
	Then, using $f^{-1} > 1$, we have for all $R \in [R_1,R_2]$,
	\[
		\int_{B(0,R_2)} \left( 2 |\nabla \varphi|^2 + v |\varphi|^2 \right) \ge \frac{b(v)}{ |\mathbb{S}^{d-1}| R^{d-1}} \left( 1- \frac{C_d R_0}{R}\right) \int_{\partial B(0,R)} |\varphi|^2.
	\]
	Now letting $U(R) = |\{ R_1 \le |\bx| \le R_2\}|^{-1} \1_{R_1\leq R \leq R_2}$, multiplying by $|\mathbb{S}^{d-1}| R^{d-1}$ on both sides, and integrating over $[R_1,R_2]$, one obtains the desired result.
\end{proof}

\section{Lower bound} \label{sec:LB}
In this section we prove the lower bound
\[
	e(\rho) \ge \frac{1}{6} b_{\mathcal M}(V) \rho^3 (1 + O(Y^\nu))\,, \quad \textrm{ when } \quad Y:=\rho b_{\cM}(V)^{3/4} \to 0\,.
\]
We follow the strategy of~\cite{LieYng-98} and estimate the energy of the box $[-L/2,L/2]^3$ by the sum of the energy on smaller boxes $[-\ell/2,\ell/2]^3$ with Neumann boundary conditions. The interaction among the boxes is discarded (here the positivity of the potential is crucial) and the number of particles in each box is controlled by a sub-additivity argument. Let us introduce the Hamiltonian
\[
	\widetilde H_{n,\ell} = \sum_{i=1}^n -\Delta_{x_i} + \sum_{1\le i < j < k \le n} \ell^2 V\left(\ell(x_i-x_j,x_i-x_k)\right)\quad \text{ on }L^2_s\!\left([-1/2,1/2]^{3N}\right),
\]
where $-\Delta$ is the Neumann Laplacian on $L^2([-1/2,1/2]^3)$ and $\ell >0$. Denoting $\mathcal U \Psi = \ell^{3n/2}\Psi(\ell \cdot)$, one has $ H_{n,\ell} = \ell^{-2} \mathcal U^* \widetilde{H}_{n,\ell} \mathcal U$ acting on $L^2_s([-\ell/2,\ell/2]^{3N})$, where we recall that $H_{n,\ell}$ is defined in~\eqref{eq:H_NL}. We have the following result.

\begin{proposition}[Energy at short length scales]\label{prop:est_H_box}
	Let
	\[
		\frac{1}{3} < \alpha < \frac{3}{5}\,, \quad \ell \sim b_{\mathcal M}(V)^{1/4} \left(\rho b_{\mathcal M}(V)^{3/4}\right)^{-\alpha} \quad \textrm{ and } \quad 0 \le n \le 10 \rho \ell^3\,.
	\]
	Then we have
	\[
		\widetilde H_{n,\ell} \ge \frac{b_{\mathcal M}(V)}{6 \ell^4} n(n-1)(n-2) + \mathcal O\! \left(\rho^3 b_{\mathcal M}(V)\ell^5 Y^{\nu}\right),
	\]
	for some constant $\nu >0$.
\end{proposition}

\begin{remark}
	By analogy to the two-body case, one can define the parameter $a = b_{\cM}(V)^{1/4}$ that plays the role of a scattering length. Using the scaling properties of $b_{\cM}(V)^{1/4}$, one can rewrite the potential
	\[
		\ell^2 V\left(\ell(x-y,x-z)\right) = \left(\frac{\ell}{a}\right)^2 \widetilde V \left(\frac{\ell}{a}(x-y,x-z)\right),
	\]
	with $b_{\cM}(\widetilde V) = 1$. In~\cite{NamRicTri-21}, the Gross--Pitaevskii regime was studied where the potential scaled like $n \widetilde V(n^{1/2}(x-y,x-z))$ for a potential $\widetilde V$ with fixed scattering energy. Assuming that there are approximately $n \simeq \rho\ell^3$ particles in the box, this corresponds to taking
	\[
		\left(\frac{\ell_{\rm GP}}{a}\right)^2 = \rho \ell_{\rm GP}^3 \iff \ell_{\rm GP} = \frac{a}{\rho a^3}\,.
	\]
	We will however consider length scales \emph{much shorter} than in the Gross--Pitaevskii regime. More precisely, we will choose $\ell \sim a (\rho a^3)^{-\alpha} \sim b_{\mathcal M}(V)^{1/4} Y^{-\alpha} $ for some $\alpha < 1$. For those length scales the gap of kinetic operator is large enough to apply the Temple inequality.
\end{remark}

\subsection{Proof of Proposition~\ref{prop:est_H_box}}

We will replace the singular potential $\ell^2 V\left(\ell\cdot\right)$ by a softer one $N^{-2} R^{-6}U(R^{-1}\cdot)$ and then use the Temple inequality to conclude. For the trapped systems in $\R^3$ considered in~\cite{NamRicTri-21}, the renormalization of the potential was implemented using~\cite[Lemma 8]{NamRicTri-21}. This result holds on the whole $\mathbb{R}^{3}$ and cannot be directly applied to our case because of the presence of the boundary. However, a very simple adaptation of it gives the following result, for which we define 
\[
	\Lambda_{\eta} = (1-\eta) [-1/2,1/2]^{3}, \quad \text{ for }\eta>0\,.
\]

\begin{lemma}[Many-body Dyson lemma in a finite box] \label{lem:gen_dyson_lemma}
	Let $0\le W \in L^\infty(\R^6)$ be supported in $B(0,R_0)$ and satisfy the symmetry~\eqref{eq:sym}. Define $U$ as in Lemma~\ref{lem:dyson_lemma}, with $\supp \, U \in \{R_1 \le |\bx| \le R_2\}$ and $U_R=R^{-6}U(R^{-1}\cdot)$ for $R>0$. Then, for all $R,\eta>0$ such that $\eta> R_2 R$ and $R_1 R > R_0$, the following holds on $L^2_s(\Lambda^N)$
	\begin{align*}
		\sum_{i=1}^N p_i^2 &+\frac{1}{6} \sum_{\substack{ 1\le i,j,k \le N \\ i\neq j \neq k \neq i } } W (x_i-x_j, x_i-x_k) \\
		&\ge \begin{multlined}[t]
			\frac{b_{\cM}(W)}{6} \left( 1 - \frac{C R_0}{R}\right) \times \\
			\times \sum_{\substack{ 1\le i,j,k \le N \\ i\neq j \neq k \neq i } } U_R(x_i-x_j, x_i-x_k) \1_{\Lambda_\eta}(x_i) \prod_{l \neq i,j,k} \theta_{2R}\left(\frac{x_i+x_j+x_k}{3}-x_l \right).
		\end{multlined}
	\end{align*}
	Here $C>0$ is a universal constant (independent of $W,R,U,N,\eta$).
\end{lemma}
\begin{proof}
	The proof is similar to that of~\cite[Lemma 8]{NamRicTri-21}, we will just recall the main steps. First, we denote $\chi_R= \1_{\{|x|\le R\}}=1- \theta_R$ for $x\in \R^3.$ For $(x_1,\ldots,x_M)\in (\R^3)^M$ and $i,j,k\in \{1,2,\ldots,M\}$ with $i\ne j\ne k\ne i$, we define
	\begin{equation} \label{eq:Fijk}
		F_{ijk} := \chi_R(x_i-x_j)\chi_R(x_i-x_k) \chi_R(x_j-x_k) \prod_{l \neq i,j,k} \theta_{2R}\left(\frac{x_i+x_j+x_k}{3} -x_l\right).
	\end{equation}
	This cut-off is such that for every $1\le i\le M$, there could be at most one pair $j,k$ such that $F_{ijk}=F_{ikj}\ne 0$. Thus we have the ``no four-body collision'' bound
	\begin{equation} \label{eq:no_4_body}
		\sum_{\substack{1\le j,k \le M\\ i\ne j \ne k \ne i}} F_{ijk} \le 2. 
	\end{equation}
	Multiplying it by $p_i$ on both sides, summing over $1\le i \le N$, we obtain
	\begin{multline*}
		\sum_{i=1}^N p_i^2 + \frac{1}{6} \sum_{\substack{1\le i,j,k \le N\\ i\ne j \ne k \ne i}} W(x_i-x_j,x_i-x_k) \\
		\ge \frac{1}{6} \sum_{\substack{1\le i,j,k \le N\\ i\ne j \ne k \ne i}} \left[ \sum_{q\in \{i,j,k\}} p_q F_{ijk} p_q + W(x_i-x_j,x_i-x_k) \right].
	\end{multline*}
	Now, we focus on the summand inside the square brackets. Let us change coordinates as in~\eqref{eq:remove-center} and apply Lemma~\ref{lem:dyson_lemma}. It gives
	\begin{multline*}
		\sum_{q\in \{i,j,k\}} p_q F_{ijk} p_q + W(x_i-x_j,x_i-x_k) \\
		\begin{aligned}[t]
			&\ge \left[ 2 \cM p_{{\bf r}_{ij}} \1_{ \{|{\bf r}_{jk}| \le R/2 \}} \cM p_{{\bf r}_{ij}} + W ({\bf r}_{jk})\right] \prod_{l \neq i,j,k} \theta_{2R}(r_i -x_l)\1_{\Lambda_{\eta}}(r_i) \\
			&\ge (1-C R_0/R) b_{\cM} (W) U_R({\bf r}_{jk}) \prod_{l \neq i,j,k} \theta_{2R}(r_i -x_l)\1_{\Lambda_{\eta}}(r_i)\,,
		\end{aligned}
	\end{multline*}
	where we used the notation ${\bf r}_{jk} = (r_j,r_k) \in \R^6$. This finishes the proof.
\end{proof}

Let $\alpha,\beta >0$ and let us define $\ell$ and $R$ via
\[
	Y= \rho b_{\mathcal M}(V)^{3/4}, \qquad \frac{b_{\mathcal M}(V)^{1/4}}{\ell} = Y^\alpha \qquad \textrm{ and } \qquad R = Y^{\beta}.
\]
We also use the notations $\mathcal T = \sum_{i=1}^n -\Delta_{x_i}$ and 
\begin{multline*}
	\mathcal W_U = \frac{b_{\mathcal M}(V)}{6 \ell^4} \left( 1 - \frac{C R_0}{R}\right) \times \\
	\times \sum_{\substack{ 1\le i,j,k \le n \\ i\neq j \neq k \neq i } } U_R(x_i-x_j, x_i-x_k) \1_{\Lambda_\eta}(x_i) \prod_{l \neq i,j,k} \theta_{2R}\left(\frac{x_i+x_j+x_k}{3}-x_l \right).
\end{multline*}
Then for $\eta = 2 R > 4 R_0 = C b_{\mathcal M}(V)^{1/4} / \ell $ and $\varepsilon >0$, applying Lemma~\ref{lem:gen_dyson_lemma} we obtain
\[
	H_n \ge \varepsilon \mathcal T + (1-\varepsilon) \mathcal W_U\,.
\]
We now use the Temple inequality. Let us consider $\varepsilon \mathcal T$ as the unperturbed Hamiltonian. Its ground state, associated to $\lambda_0 = 0$, is $\Psi_0 = 1$ and its second eigenvalue is $\lambda_1 = \varepsilon \pi$. The perturbation $(1-\varepsilon) \mathcal W_U$ is non-negative thence, as long as
\begin{equation} \label{eq:cond_Temple}
	\lambda_1 - \lambda_0 > \langle \Psi_0, (1-\varepsilon) \mathcal W_U \Psi_0 \rangle\,,
\end{equation}
it holds by the Temple inequality that
\begin{equation} \label{eq:Temple}
	E_0 \ge \lambda_0 + \langle \Psi_0, (1-\varepsilon) \mathcal W_U \Psi_0 \rangle - (1-\varepsilon)^2 \frac{\langle \Psi_0,\mathcal W_U^2 \Psi_0 \rangle - \langle \Psi_0,\mathcal W_U \Psi_0 \rangle^2}{\lambda_1 - \lambda_0 - (1-\varepsilon) \langle \Psi_0, \mathcal W_U \Psi_0 \rangle}\,.
\end{equation}
We need to control $ \langle \Psi_0, \mathcal W_U \Psi_0\rangle $ and $ \langle \Psi_0, \mathcal W_U^2 \Psi_0\rangle$ and to obtain a good lower bound for $(1-\varepsilon)\langle \Psi_0, \mathcal W_U \Psi_0\rangle$. A simple computation yields
\[
	\langle \Psi_0,\mathcal W_U \Psi_0 \rangle \le C \frac{b_{\mathcal M}(V)}{\ell^4} n^3 \le C Y^{3-5\alpha}.
\]
Therefore, taking $\varepsilon \gtrsim Y^{3-5\alpha}$, we ensure that~\eqref{eq:Temple} is valid.
In order to estimate $ \langle \Psi_0, \mathcal W_U^2 \Psi_0\rangle$, we use the pointwise bound
\[
	\sum_{\substack{ 1\le i,j,k \le n \\ i\neq j \neq k \neq i } } U_R(x_i-x_j, x_i-x_k) \1_{\Lambda_\eta}(x_i) \prod_{l \neq i,j,k} \theta_{2R}\left(\frac{x_i+x_j+x_k}{3}-x_l \right) \le C \|U\|_{L^\infty} R^{-6} n\,,
\]
from which we obtain 
\[
	\langle \Psi_0, \mathcal W_U^2 \Psi_0\rangle \le C \left(n R^{-6} \frac{b_{\mathcal M}(V)}{\ell^4}\right)^2 \le C Y^{2 +2\alpha - 12 \beta}.
\]
Consequently, the last term in~\eqref{eq:Temple} is estimated by
\begin{align*}
(1-\varepsilon)^2 \frac{\langle \Psi_0,\mathcal W_U^2 \Psi_0 \rangle - \langle \Psi_0,\mathcal W_U \Psi_0 \rangle^2}{\lambda_1 - \lambda_0 - (1-\varepsilon) \langle \Psi_0, \mathcal W_U \Psi_0 \rangle} 
	&\le C \varepsilon^{-1} \frac{Y^{2 +2\alpha - 12 \beta}}{1 - C \varepsilon^{-1}Y^{3-5\alpha}} \\
	&\le C \varepsilon^{-1} \rho^3 b_{\mathcal M}(V) \ell^5 Y^{7\alpha - 12 \beta - 1}.
\end{align*}
We now estimate by below the second term in~\eqref{eq:Temple}. Using that 
\[
	\prod_{l \neq i,j,k} \theta_{2R}\left(\frac{x_i+x_j+x_k}{3}-x_l \right) \ge 1 - \sum_{l \neq i,j,k} \left( 1 - \theta_{2R}\left(\frac{x_i+x_j+x_k}{3}-x_l \right)\right),
\]
we obtain
\begin{align*}
\langle \Psi_0,\mathcal W_U \Psi_0 \rangle 
	&\ge \frac{b_{\mathcal M}(V)}{6 \ell^4} n(n-1)(n-2) \left(1 - C nR^3 \right) (1-\eta)\left( 1 - \frac{C b_{\mathcal M}(V)^{1/4}}{\ell R}\right) \\
	&\ge \frac{b_{\mathcal M}(V)}{6 \ell^4} n(n-1)(n-2) \left(1 - C \left(Y^{1-3(\alpha-\beta)} + Y^{\beta} + Y^{\alpha-\beta}\right)\right). \\	
\end{align*}
Thus the Temple inequality~\eqref{eq:Temple} leads to
\begin{multline*}
	E_0 \ge \frac{b_{\mathcal M}(V)}{6 \ell^4} n(n-1)(n-2) \\
	- C \rho^3 b_{\mathcal M}(V) \ell^5 \left(Y^{1-3(\alpha-\beta)} + Y^{\beta} + Y^{\alpha-\beta} + \varepsilon + \varepsilon^{-1} Y^{7\alpha - 12 \beta - 1} \right).
\end{multline*}
We now take $\alpha$ and $\beta$ to satisfy the conditions
\[
	\frac{1}{3} < \alpha < \frac{3}{5} \quad \textrm{ and } \quad \alpha - \frac{1}{3} < \beta < \frac{7\alpha - 1}{12}\,,
\]
for which the range for $\beta$ is non empty. With this choice of parameters and choosing $\varepsilon = Y^{(7\alpha-12\beta -1)/2} \gg Y^{3 - 5\alpha}$, we conclude the proof of Proposition~\ref{prop:est_H_box}.
\qed

\subsection{Proof of the lower bound in Theorem~\ref{theo1}}

Let us divide $\Omega = [-L/2,L/2]^3$ in $M(N,L)^3$ boxes of size $\ell (N,L)>0$, where $ M \in \mathbb{N}$. We therefore have the relation $L = M \ell$. More precisely, we parametrize $M = \lfloor L Y^\alpha \rfloor$, for some $1/3 < \alpha < 3/5$, so that
\[
	\lim_{\substack{N \to \infty \\ N /L^3 \to \rho}}\ell (L,N) = \frac{b_{\cM}(V)^{1/4}}{Y^{\alpha}}\,.
\]
Let $(B_i)_{1\le i \le M^3}$ be the boxes and let us use the notation $\overline{A} = \mathbb{R}^{3} \setminus A$. For $\Psi \in L^2_s(\Omega^N)$ and $1\le k \le N$, we define 
\[
	c_{k} = \frac{1}{M^3} \int_{\Omega^N} \sum_{i=1}^{M^3} \delta_{k, \sum_{n=1}^N \mathds{1}_{B_{i}}(x_n)} |\Psi|^2 = \frac{1}{M^3} \binom{N}{k} \sum_{i=1}^{M^3}\int_{B_i^{k} \times \overline{B_i}^{N-k}} |\Psi|^2 \,,
\]
where the second equality is obtained by expanding $1 = ( \mathds{1}_{B_i} + \mathds{1}_{\overline{B_i}})^{N}$ and using the symmetry of $\Psi$. The quantity $M^3 c_k$ is the expected number of boxes containing exactly $k$ particles. Using that $\int |\Psi|^2 = 1$, we obtain
\begin{align*}
	\sum_{k=0}^N c_k &= \frac{1}{M^3} \int_{\Omega^N} \sum_{i=1}^{M^3} \sum_{k=0}^N \delta_{k, \sum_{n=1}^N \mathds{1}_{B_{i}}(x_n)} |\Psi|^2 = \int_{\Omega^N} |\Psi|^2 = 1
	\intertext{and}
	\sum_{k=0}^N k c_k &= \frac{1}{M^3} \int_{\Omega^N} \sum_{i=1}^{M^3} k \delta_{k, \sum_{n=1}^N \mathds{1}_{B_{i}}(x_n)} |\Psi|^2 = \frac{1}{M^3} \int_{\Omega^N} N |\Psi|^2 = \rho \ell^{3}.
\end{align*}

Let us now deal with the energy. Using again the symmetry of $\Psi$ we obtain
\begin{multline}
	\langle \Psi, H_N \Psi \rangle \\
	\begin{aligned}[b]
		&\begin{multlined}[t][0.87\textwidth]
			= N \int_{\Omega^N} |\nabla_{1} \Psi (x_1,\dots,x_N)|^2 \\
			+ \frac{N(N-1)(N-2)}{6} \int_{\Omega^N} V(x_1,x_2,x_3) |\Psi (x_1,\dots,x_N)|^2
		\end{multlined} \\
		&\ge \sum_{i=1}^{M^3} N \int_{B_i \times \Omega^{N-1}} |\nabla_{1} \Psi |^2 + \frac{N(N-1)(N-2)}{6} \int_{B_i^3 \times \Omega^{N-3}} V(x_1,x_2,x_3) |\Psi |^2 \\
		&\begin{multlined}[t][0.87\textwidth]
			\ge \sum_{i=1}^{M^3} \sum_{k=0}^{N-1} N \binom{N-1}{k} \int_{B_i^{1+k} \times \overline{B_i}^{N-1-k}} |\nabla_{1} \Psi |^2 \\
			+ \sum_{i=1}^{M^3} \sum_{k=0}^{N-3} \frac{N(N-1)(N-2)}{6} \binom{N-3}{k} \int_{B_i^{3+k} \times \overline{B_i}^{N-3-k}} V(x_1,x_2,x_3) |\Psi |^2
		\end{multlined} \\
		&\begin{multlined}[t][0.87\textwidth]
			\ge \sum_{i=1}^{M^3} \sum_{k=0}^{N-1} k \binom{N}{k} \int_{B_i^{k} \times \overline{B_i}^{N-k}} |\nabla_{1} \Psi |^2 \\
			+ \sum_{i=1}^{M^3} \sum_{k=0}^{N-3} \frac{k(k-1)(k-2)}{6} \binom{N}{k} \int_{B_i^{k} \times \overline{B_i}^{N-k}} V(x_1,x_2,x_3) |\Psi |^2
		\end{multlined} \\
		&\ge M^3 \sum_{k=0}^N c_k E(\ell, k)\,,
	\end{aligned}
\end{multline}
where $E(\ell,k) = \inf \sigma(H_{\ell,k})$ and $V(x_1,x_2,x_3):= V(x_1-x_2,x_1-x_3)$ is symmetric in $x_1,x_2,x_3$. Here, we used that
\begin{multline*}
	k \int_{B_i^{k} \times \overline{B_i}^{N-k}} |\nabla_{1} \Psi |^2 + \frac{1}{6} k(k-1)(k-2) \int_{B_i^{k} \times \overline{B_i}^{N-k}} V(x_1,x_2,x_3) |\Psi |^2 \\
	= \tr H^{B_i}_{k} \gamma_{i,k} \ge E(\ell, k) \tr \gamma_{i,k}\,,
\end{multline*}
where $H^{B_i}_{k} $ is the translation of $H_{k,\ell}$ to the box $B_i$ and the $\gamma_{i,k}$'s, defined as
\begin{multline*}
	\gamma_{i,k} (x_1,\dots,x_k; y_1,\dots,y_k) \\
	:= \binom{N}{k}\int_{B_i^{k} \times \overline{B_i}^{N-k}} \Psi(x_1,\dots,x_k,z_{k+1},\dots,z_{N}) \overline{\Psi(y_1,\dots,y_k,z_{k+1},\dots,z_{N})} \dd \bz \,,
\end{multline*}
satisfy $\sum_{i=1}^{M^3} \tr \gamma_{i,k} = M^3 c_k$. We now use that for $k \le 10 \rho \ell^3$, by Proposition~\ref{prop:est_H_box}, we have
\[
	E(\ell, k) \ge \frac{b_{\mathcal M}(V)}{6 \ell^6} k(k-1)(k-2) + o\!\left(\rho^3 b_{\mathcal M}(V) \ell^3 \right)\,.
\]
By the convexity of $t \mapsto t(t-1)(t-2)$ and denoting $x = \sum_{k \le 10\rho\ell^3} k c_k \le \rho \ell^3$,
\begin{equation} \label{eq:exp_E_1}
	\sum_{k \le 10 \rho \ell^3} c_k E(\ell, k) \ge \frac{b_{\mathcal M}(V)}{6 \ell^6} x (x-1)(x-2) + \mathcal O \!\left(\rho^3 b_{\mathcal M}(V) \ell^3 Y^{\nu} \right)
\end{equation}
holds for some $\nu >0$ as in Proposition~\ref{prop:est_H_box}.
On the other hand, using that $E(\ell, k+k') \ge E(\ell, k) + E(\ell, k')$, we obtain
\begin{equation} \label{eq:exp_E_2}
	\begin{aligned}[b]
		\sum_{k > 10 \rho \ell^3} c_k E(\ell, k) 
		&\ge \frac{1}{2} \frac{k}{10 \rho \ell^3} E\!\left(\ell, 10 \rho \ell^3 \right) \\
		&\ge \frac{\rho \ell^3 - x}{12} \left\{ \left(10\rho\ell^3 -1\right) \left(10\rho\ell^3 - 2 \right) + \mathcal O\!\left(\rho^2 b_{\mathcal M}(V) Y^{\nu} \right)\right\}.
	\end{aligned}
\end{equation}
Summing up~\eqref{eq:exp_E_1} and~\eqref{eq:exp_E_2}, we notice that the minimum is attained for $x = \rho\ell^3$. Recall that we choose $\ell \sim a Y^{-\alpha}$ so that $\rho \ell^3 \sim Y^{1-3\alpha} \gg 1$ for $\alpha > 1/3$. We obtain
\begin{align*}
\frac{\langle \Psi, H_N \Psi \rangle}{L^3} 
	&\ge \frac{M^3}{L^3} \left(\frac{b_{\mathcal M}(V)}{6 \ell^6} \rho \ell^3 \left(\rho \ell^3-1\right) \left(\rho \ell^3-2\right) + \mathcal O\!\left(\rho^3 b_{\mathcal M}(V) \ell^3 Y^{\nu} \right)\right) \\
	&\ge \frac{1}{6} \rho^3 b_{\mathcal M}(V) \left(1 - 3 Y^{3\alpha-1} + \mathcal O \!\left( Y^{\nu} \right) \right),
\end{align*}
where we used that $L = \ell M$.
\qed

\section{Upper bound}\label{sec:UB}

Let us define
\[
	\widetilde H_{n}^{\rm Dir} := \sum_{i=1}^n -\Delta^{\rm Dir}_{x_i} + \sum_{1\le i < j < k \le n} n W \left(n^{1/2}(x_i-x_j,x_i-x_k)\right)
\]
acting on $L_s^2([-1/2,1/2]^{3n})$, where $-\Delta^{\rm Dir}$ is the Laplacian with Dirichlet boundary condition. A simple adaptation of~\cite{NamRicTri-21} gives the following result.
\begin{lemma}
	\label{lem:upper_bound}
	We have,
	\begin{multline*}
		\inf_{\|\Psi\|_{L^2}=1} \frac{\langle \Psi, \widetilde H_{n}^{\rm Dir} \Psi\rangle}{n} \le \frac{1}{6} b_{\cM}(W) \left( 1+ C b_{\cM}(W)^{-1/2} \right) \\
		+ C n^{-1/3} (b_{\cM}(W) +1)^{12} \left(1 +\|W\|_{L^1} + \|W\|_{L^\infty}\right)^2,
	\end{multline*}
	where the constant $C$ does not depend on $W$.
\end{lemma}
\begin{proof}
	Following the proof of~\cite[Corollary 19]{NamRicTri-21}, we have for all $\varphi \in H^2_0([-1/2,1/2]^3)$ the bound
	\begin{multline*}
		\inf_{\|\Psi\|_{L^2}=1} \frac{\langle \Psi, \widetilde H_{n}^{\rm Dir} \Psi\rangle}{n} \\
		\le \left(\int_{\Omega} |\nabla \varphi|^2 +\frac{1}{6} b_{\cM}(W)\int_{\Omega} |\varphi|^6 \right) + C(\|\varphi\|_{L^2}, \|\varphi\|_{L^\infty})n^{-1/3}(b_{\cM}(W) +1)^3 \times \\
		\times \|\varphi\|_{H^2}^6 \left(1 +\|W\|_{L^1} + \|W\|_{L^\infty}\right)^2.
	\end{multline*}
	Now, for $\varepsilon >0$ small enough, let us choose $\varphi \in H_0^2([-1/2,1/2]^3)$ such that $0\le \varphi \le 1$, $\|\varphi\|_{L^2}=1$, $\varphi \equiv 1$ on $(1-\varepsilon)[-1/2,1/2]^3$, $|\nabla \varphi| \le C \varepsilon^{-1}$ and $|\Delta \varphi| \le C \varepsilon^{-2}$. Then we have $\|\nabla \varphi\|_{L^2}^2 \le C \varepsilon^{-1}$ and $\|\varphi\|_{H^2} \le C (\varepsilon^{-3}+1)$. Choosing $\varepsilon = b_{\cM}(W)^{-1/2}$, we obtain the claim.
\end{proof}

Let us introduce again the effective scattering length $a = b_{\cM}(V)^{1/4}$, let $\rho>0$ and let us denote $Y = \rho a^3$ the diluteness parameter. We will create a trial state made of collections of Dirichlet minimizers of boxes of size $\ell \simeq K \ell_{GP}$ slightly bigger than the Gross--Pitaevskii length scale, that is $K \gg 1$. Hence the localization error will be subleading. We parametrize $K = Y^{-\alpha}$ for some $\alpha >0$ that we will take sufficiently small, and we define $n = \lfloor K^3 Y^{-2} \rfloor \sim Y^{3\alpha -2}$ the number of particles that we put in each box. We consider boxes of size $\ell := a n^{1/2} K^{-1/2} \sim a K Y^{-1}$. With this choice of parameters, we have
\[
	n\ell^{-3} = \rho \left(1 + \mathcal O\!\left(Y^{2-3\alpha} \right) \right),
\]
as $Y = \rho a^3 \to 0$. Moreover denoting again $\mathcal U \Psi = \ell^{3n/2}\Psi(\ell \cdot)$, in the sense of quadratic forms on $\bigvee^N_s H^1_0([0,M(\ell + R_0)]^{3})$ one has $ H_{n,\ell} = \ell^{-2} \mathcal U^* \widetilde{H}^{\rm Dir}_{n,\ell} \mathcal U$, with 
\[
	W(x,y) = a^2 K^{-1} V\!\left( a K^{-1/2}(x,y) \right)
\]
and where we recall that $H_{n,\ell}$ is defined in~\eqref{eq:H_NL}. Note that $b_{\cM}(W) = K^2$.

For $M\ge 1$ and $z \in \llbracket 0,M-1\rrbracket^3$, we denote $B_z = (\ell+R_0) z + [0,\ell]^3$ and $\Psi_z$ the normalized ground state of $\mathcal U^* \widetilde{H}^{\rm Dir}_{n,\ell} \mathcal U$ in $B_z$. Thanks to the Dirichlet boundary condition, we can extend $\Psi_z$ by $0$ outside $B_z$. Then, taking $N = n M^3$, the bosonic state $\Psi := \bigvee_{z \in \llbracket 0,M-1\rrbracket^3} \Psi_z \in L^2_s([0,M(\ell + R_0)]^{3N})$ is normalized and since the boxes are at distance $R_0$ apart from each other, its reduced density matrices satisfy
\[
	\Gamma^{(1)}_{\Psi} = \sum_{z} \Gamma^{(1)}_{\Psi_z} \quad \textrm{ and } \quad \tr V(x-y,x-z) \Gamma^{(3)}_{\Psi} = \sum_{z} \tr V(x-y,x-z) \Gamma^{(3)}_{\Psi_z}\,.
\]
Here we have used the notations
\begin{align*}
	&\begin{multlined}[t][\textwidth]
		\Gamma^{(k)}_{\Psi}(x_1,\dots,x_k; y_1,\dots,y_k) \\
		= \binom{N}{k} \int_{[0,M(\ell + R_0)]^{3(N-k)}} \!\!\!\! \Psi(x_1,\dots,x_k,z_{k+1},\dots,z_{N}) \overline{\Psi}(y_1,\dots,y_k,z_{k+1},\dots,z_{N}) \d \bz
	\end{multlined}
	\intertext{and}
	&\begin{multlined}[t][\textwidth]
		\Gamma^{(k)}_{\Psi_z}(x_1,\dots, x_k; y_1,\dots,y_k) \\
		= \binom{n}{k} \int_{B_z^{(n-k)}} \Psi_z(x_1,\dots,x_k,z_{k+1},\dots,z_{n}) \overline{\Psi}_z(y_1,\dots,y_k,z_{k+1},\dots,z_{n}) \d \bz \,.
	\end{multlined}
\end{align*}
Therefore, combining the above with Lemma~\ref{lem:upper_bound}, we obtain
{\allowdisplaybreaks
\begin{align*}
	&\frac{\langle \Psi, H_{N,L} \Psi \rangle}{M^3 (\ell + R_0)^3} = \frac{1}{M^3 (\ell + R_0)^3}\sum_z \tr (-\Delta) \Gamma^{(1)}_{\Psi_z} + \tr V(x-y,x-z) \Gamma^{(3)}_{\Psi_z} \\
	&\begin{multlined}[t][0.975\textwidth]
			\le \frac{1}{M^3 (\ell + R_0)^3}\sum_z \frac{n}{\ell^2} \left( \frac{1}{6} b_{\cM}(W) + C b_{\cM}(W)^{1/2}\vphantom{C n^{-1/3} (b_{\cM}(W) +1)^{12} \left(1 +\|W\|_{L^1} + \|W\|_{L^\infty}\right)^2} \right. \\
			\left. +\ C n^{-1/3} (b_{\cM}(W) +1)^{12} \left(1 +\|W\|_{L^1} + \|W\|_{L^\infty}\right)^2\vphantom{\frac{1}{6} b_{\cM}(W) + C b_{\cM}(W)^{1/2}} \right)
		\end{multlined}\\
	&\begin{multlined}[t][0.975\textwidth]
			\le \frac{1}{(1+ R_0 a^{-1} Y^{1-2\alpha})^3} \frac{n}{\ell^3} \frac{1}{\ell^2} \left( \frac{1}{6} K^2 + CK\vphantom{C \frac{K^{24}}{n^{1/3}} \left(1+K^2 a^{-4}\|V\|_{L^1} + K^{-1}a^2 \|V\|_{L^\infty} \right)^2} \right. \\
			\left. +\ C \frac{K^{24}}{n^{1/3}} \left(1+K^2 a^{-4}\|V\|_{L^1} + K^{-1}a^2 \|V\|_{L^\infty} \right)^2 \vphantom{\frac{1}{6} K^{2}}\right)
		\end{multlined}\\
	&\le \frac{1}{6}\rho^3 b_{\cM}(V) \left(1 + C_{V} \left(Y^{\alpha} + Y^{2/3 - 27 \alpha} + Y^{2-3\alpha} + Y^{1-2\alpha}\right) \right),
\end{align*}
}%
where $C_{V}$ depends only on $R_0 a^{-1}, \|V\|_{L^1} a^{-4}$, and $ \|V\|_{L^\infty} a^2$ (recall $a:=b_{\cM}(V)^{1/4}$). Moreover we used that, with our choice of $\ell$ and $n$, we have
\[
	\frac{K^2}{\ell^2} = \frac{\rho^2 a^4}{1 + \lfloor K^3 Y^{-2}\rfloor / K^3 Y^{-2} } = \rho^2 a^4 \left(1 + \mathcal O\!\left(Y^{2-3\alpha} \right) \right).
\]
Taking $0 < \alpha < 2/81$ finishes the proof of the upper bound. \qed


\bigskip
\raggedbottom


\begin{thebibliography}{10}
	\bibitem{Wieman-Cornell-95}
	{\sc M.~H.~Anderson, J.~R.~Ensher, M.~R.~Matthews, C.~E.~Wieman, and E.~A.~Cornell}, {\em Observation of {B}ose--{E}instein condensation in a dilute atomic vapor}, Science, 269 (5221) (1995), pp.~198--201.

	\bibitem{BCS-21}
	{\sc G.~Basti, S.~Cenatiempo, and B.~Schlein}, {\em A new second order upper bound for the ground state energy of dilute {B}ose gases}, Forum Math. Sigma, 9 (2021), E74.

	\bibitem{Bogoliubov-47b}
	{\sc N.~N.~Bogoliubov}, {\em On the theory of superfluidity}, J.~Phys. (USSR), 11 (1947), pp.~23--32.

	\bibitem{Bose-24}
	{\sc S.~N.~Bose.}, {\em {P}lancks {G}esetz und {L}ichtquantenhypothese}, Z. Phys., 26 (1924), pp.~178--181.

	\bibitem{BMZ-07}
	{\sc H.~P.~Buchler, A.~Micheli, and P.~Zoller}, {\em Three-body interactions with cold polar molecules}, Nature Physics 3 (2007), pp.~726--731.
	
	\bibitem{Chen-12}
	{\sc X.~Chen}, {\em{Second order corrections to mean field evolution for weakly interacting bosons in the case of three-body interactions}}, Arch. Ration. Mech. Anal., 203 (2012), pp.~455--497.

	\bibitem{CheHol-19}
	{\sc X.~Chen and J.~Holmer}, {\em The derivation of the $\mathbb{T}^3$ energy-critical {NLS} from quantum many-body dynamics}, Invent. math., 217 (2019), pp.~433--547.

	\bibitem{ChePav-11}
	{\sc T.~Chen and N.~Pavlović}, { \em{The quintic NLS as the mean field limit of a Boson gas with three-body interactions}}, J.~Func. Anal., 260 (2011), pp.~959--997.

	\bibitem{Ketterle-95}
	{\sc K.~B.~Davis, M.~O.~Mewes, M.~R.~Andrews, N.~J.~van~Druten, D.~S.~Durfee, D.~M.~Kurn, and W.~Ketterle}, {\em {B}ose--{E}instein {C}ondensation in a {G}as of {S}odium {A}toms}, Phys. Rev. Lett., 75 (1995), pp.~3969--3973.

	\bibitem{Dyson-57}
	{\sc F.~J.~Dyson}, {\em Ground state energy of a hard-sphere gas}, Phys. Rev., 106 (1957), pp.~20--26.

	\bibitem{Einstein-24}
	{\sc A.~Einstein}, {\em Quantentheorie des einatomigen idealen {G}ases}, Sitzungsber. Preu{{\ss}}. Akad. Wiss., Phys.-Math. Kl., 1924 XXII (1924), pp.~261--267.

	\bibitem{Einstein-25}
	\leavevmode\vrule height 2pt depth -1.6pt width 23pt, {\em Quantentheorie des einatomigen idealen {G}ases. {Z}weite abhandlung.}, Sitzungsber. Preu{{\ss}}. Akad. Wiss., Phys.-Math. Kl., 1925 I (1925), pp.~3--14.

	\bibitem{FouSol-20}
	{\sc S.~Fournais and J.~P.~Solovej}, {\em The energy of dilute {B}ose gases}, Ann. of Math., 192 (2020), pp.~893--976.

	\bibitem{FouSol-21}
	\leavevmode\vrule height 2pt depth -1.6pt width 23pt, {\em The energy of dilute {B}ose gases {II}: The general case}, arXiv:2108.12022.

	\bibitem{HY-57}
	{\sc K.~Huang and C.~N.~Yang}, {\em Quantum-mechanical many-body problem with hard-sphere interaction}, Phys. Rev., 105 (1957), pp.~767--775.

	\bibitem{Landau-41}
	{\sc L.~Landau}, {\em Theory of the superfluidity of helium {II}}, Phys. Rev., 60 (1941), pp.~356--358.
	
	\bibitem{Lee-20}
	{\sc J.~Lee}, {\em{Rate of convergence towards mean-field evolution for weakly interacting bosons with singular three-body interactions.}}, Preprint arXiv:2006.13040, (2020).

	\bibitem{LHY-57}
	{\sc T.~D.~Lee, K.~Huang, and C.~N.~Yang}, {\em Eigenvalues and eigenfunctions of a {B}ose system of hard spheres and its low-temperature properties}, Phys. Rev., 106 (1957), pp.~1135--1145.
	
	\bibitem{LiYao-21}
	{\sc Y.~Li and F.~Yao}, {\em{Derivation of the nonlinear Schr{\"o}dinger equation with a general nonlinearity and Gross--Pitaevskii hierarchy in one and two dimensions}}, J.~Math. Phys., 62 (2021), p.~021505.

	\bibitem{LieSei-02}
	{\sc E.~H.~Lieb and R.~Seiringer}, {\em Proof of {B}ose--{E}instein condensation for dilute trapped gases}, Phys. Rev. Lett., 88 (2002), p.~170409.

	\bibitem{LieSei-06}
	\leavevmode\vrule height 2pt depth -1.6pt width 23pt, {\em Derivation of the {G}ross--{P}itaevskii equation for rotating {B}ose gases}, Commun. Math. Phys., 264 (2006), pp.~505--537.

	\bibitem{LSSY}
	{\sc E.~H.~Lieb, R.~Seiringer, J.~P.~Solovej, and J.~Yngvason}, {\em The Mathematics of the Bose Gas and its Condensation}, Birkh\"auser-Verlag, Basel, 2005, pp.~viii+208.

	\bibitem{LieSeiYng-00}
	{\sc E.~H.~Lieb, R.~Seiringer, and J.~Yngvason}, {\em Bosons in a trap: A rigorous derivation of the {G}ross--{P}itaevskii energy functional}, Phys. Rev. A, 61 (2000), p.~043602.

	\bibitem{LieYng-98}
	{\sc E.~H.~Lieb and J.~Yngvason}, {\em Ground state energy of the low density {B}ose gas}, Phys. Rev. Lett., 80 (1998), pp.~2504--2507.
	
	\bibitem{Mas-03}
	{\sc E.~M.~Mas, R.~Bukowski and K.Szalewicz}, {\em{Ab initio three-body interactions for water.~II. Effects on structure and energetics of liquid.}} J.Chem. Phys., 118 (2003), pp.~4404--4413.

	\bibitem{NamRouSei-16}
	{\sc P.~T.~Nam, N.~Rougerie, and R.~Seiringer}, {\em Ground states of large bosonic systems: The {G}ross--{P}itaevskii limit revisited}. Analysis \& PDE 9 (2016), pp.~459--485.

	\bibitem{NamRicTri-21}
	{\sc P.~T.~Nam, J.~Ricaud, and A.~Triay}, {\em The condensation of a trapped dilute {B}ose gas with three-body interactions}, arXiv:2110.08195.
	
	\bibitem{NamSal-20}
	{\sc P.~T.~Nam and R.~Salzmann}, {\em {D}erivation of {3D} {E}nergy-{C}ritical {N}onlinear {S}chr{\"o}dinger {E}quation and {B}ogoliubov {E}xcitations for {B}ose {G}ases}, Commun. Math. Phys., 375 (2020), pp.~495--571.
	
	\bibitem{Onsager-49}
	{\sc L.~Onsager}, {\em Statistical hydrodynamics}, Il Nuovo Cimento (1949), 6(2), pp.~279--287.

	\bibitem{Petrov-14}
	{\sc D.~S.~Petrov}, {\em Three-body interacting bosons in free space}, Phys. Rev. Lett., 112 (2014), p.~103201.

	\bibitem{Ruelle}
	{\sc D.~Ruelle}, {\em Statistical mechanics. Rigorous results}, {Singapore: World Scientific. London: Imperial College Press}, 1999.

	\bibitem{YauYin-09}
	{\sc H.-T.~Yau and J.~Yin}, {\em The second order upper bound for the ground energy of a {B}ose gas}, J.~Stat. Phys., 136 (2009), pp.~453--503.
	
	\bibitem{Yuan-15}
	{\sc J.~Yuan}, {\em{Derivation of the Quintic NLS from many-body quantum dynamics in $\mathbb{T}^2$}}, Comm. Pure Appl. Anal., 14 (2015), p.~1941
\end{thebibliography}
\end{document}